\documentclass[11pt]{article}
\usepackage{amsmath,amsbsy,amsfonts,amssymb,amsthm,commath}

\usepackage{hyperref}
\usepackage{algorithm,algorithmic}
\usepackage{cleveref}
\usepackage{ifmtarg}
\usepackage[margin=1in]{geometry}
\usepackage{bbm}
\usepackage{authblk}
\usepackage{enumitem}

\newtheorem{theorem}{Theorem}[section]
\newtheorem{lemma}[theorem]{Lemma}

\newtheorem{claim}[theorem]{Claim}

\newtheorem{corollary}[theorem]{Corollary}
\newtheorem{remark}[theorem]{Remark}

\newtheorem{definition}[theorem]{Definition}

\makeatletter
\newcommand{\toprule}{\hrule height.8pt depth0pt \kern2pt}
\newcommand{\midrule}{\kern2pt\hrule\kern2pt}
\newcommand{\bottomrule}{\kern2pt\hrule\relax}
\newcommand{\algcaption}[2][]{%
  \refstepcounter{algorithm}%
  \@ifmtarg{#1}
    {\addcontentsline{loa}{figure}{\protect\numberline{\thealgorithm}{\ignorespaces #2}}}
    {\addcontentsline{loa}{figure}{\protect\numberline{\thealgorithm}{\ignorespaces #1}}}%
  \toprule
  \textbf{\fname@algorithm~\thealgorithm}\ #2\par
  \midrule
}
\makeatother

\newcommand\bbF{\ensuremath{\mathbb{F}}} 

\newcommand\cM{\ensuremath{\mathcal{M}}}

\newcommand{\E}{\mathbb{E}}

\newcommand{\Cmu}{\cC(\sig,\delta)}
\newcommand{\sig}{\sigma}

\newcommand{\cC}{\mathcal{C}}

\newcommand{\cG}{\mathcal{G}}

\newcommand{\acs}{\mathcal{A}}

\newcommand{\F}{{\mathbb{F}}}
\newcommand{\one}{{\mathbbm{1}}}
\DeclareMathOperator{\poly}{poly}

\newcommand{\iS}{\one_S}

\newcommand{\bP}{\mathbf{P}}

\newcommand\Enc{\ensuremath{\mathsf{Enc}}}
\newcommand\Dec{\ensuremath{\mathsf{Dec}}}
\newcommand\MLDec{\ensuremath{\mathsf{Dec}^{\mathsf{ML}}}}
\newcommand\supp{\ensuremath{\operatorname{supp}}}

\usepackage{color}

\newcommand{\mnote}[1]{ \marginpar{\tiny\bf
            \begin{minipage}[t]{0.5in}
              \raggedright #1
           \end{minipage}}}

\renewcommand{\mnote}[1]{}

\title{Coding Sets with Asymmetric Information}
\author[1]{Alexandr Andoni}
\author[2]{Javad Ghaderi}
\author[1]{Daniel Hsu}
\author[1]{Dan Rubenstein}
\author[1]{Omri Weinstein}
\affil[1]{Department of Computer Science, Columbia University}
\affil[2]{Department of Electrical Engineering, Columbia University}

\begin{document}
\maketitle
{\def\thefootnote{}\footnotetext{E-mail: \texttt{\{andoni@cs,jghaderi@ee,djhsu@cs,danr@cs,omri@cs\}.columbia.edu}}}

\begin{abstract}
We study the following one-way asymmetric transmission problem, also a variant of model-based compressed sensing: a resource-limited encoder has to report a small set $S$ from a universe of $N$ items to a more powerful decoder (server).  The distinguishing feature is \emph{asymmetric information}: the subset $S$ is comprised of i.i.d.~samples from a prior distribution $\mu$, and $\mu$ is only known to the \emph{decoder}.  The goal for the encoder is to encode $S$ \emph{obliviously}, while achieving the information-theoretic bound of $|S| \cdot H(\mu)$, i.e., the Shannon entropy bound.

We first show that any such compression scheme must be {\em randomized}, if it gains non-trivially from the prior $\mu$.  This stands in contrast to the symmetric case (when both the encoder and decoder know $\mu$), where the Huffman code provides a near-optimal \emph{deterministic} solution.  On the other hand, a rather simple argument shows that, when $|S|=k$, a \emph{random} linear code achieves near-optimal communication rate of about $k\cdot H(\mu)$ bits. Alas, the resulting scheme has prohibitive \emph{decoding time}: about ${N\choose k} \approx (N/k)^k$.

Our main result is a \emph{computationally efficient} and \emph{linear} coding scheme, which achieves an $O(\lg\lg N)$-competitive communication ratio compared to the optimal benchmark, and runs in $\text{poly}(N,k)$ time. Our ``multi-level'' coding scheme uses a combination of hashing and syndrome-decoding of Reed-Solomon codes, and relies on viewing the (unknown) prior $\mu$ as a rather small convex combination of uniform (``flat'') distributions.
\end{abstract}

\thispagestyle{empty}
\newpage
\setcounter{page}{1}

\section{Introduction}
\label{sec:motivation}

We study the problem of coding a set with \emph{asymmetric}
information, defined as follows.  There is a universe $[N] :=
\cbr{1,2,\dotsc,N}$ of $N$ items, and the encoder's task is to
transmit a subset $S\subset [N]$ using an $m$-bit message so that a
decoder can reconstruct the set $S$ efficiently. In our setup, the
decoder has a {\em prior distribution} $\sigma$ over the sets $S$ that
may be sent, which is not available to the encoder.  The main goal is
to design compression schemes that (1) obtain communication rate as
close as possible to the information-theoretic minimum, namely the
(Shannon) entropy bound with respect to the distribution $\sigma$, and
(2) are computationally efficient.

This problem is the one-way communication version of the asymmetric
transmission problem \cite{adler1998protocols}, as well as a type of
model-based compressed sensing. While we expand on these a little
below, for now we note that the standard asymmetric transmission
problem is two-way, with the decoder sending much more information to
the encoder. Here we seek to eliminate this inefficiency, in the
setting of communicating a set $S$. One can envision many scenarios
where it is imperative to eliminate an expensive down-link from
decoder to encoder; we give one such scenario for designing very light
communication protocols for tracking ultra-low-power devices in
Internet-of-Things environments. Here, a common task is for a set of
such devices to communicate their identities to a router (e.g., an
entry point of a physical region)
\cite{gorlatova2009challenge,chen2016maximizing,buettner2011dewdrop}. Since
the devices are low power, the main goal is to minimize their total
communication costs. The communication can be further improved using
some side information, in particular a prior distribution on which
devices are more likely to be present (i.e., which sets are
more likely to be sent). However, the side information is typically
asymmetric: the prior is specific to the decoding router, or uses
statistics that are not known to or are too expensive to maintain by
the devices (see the discussion in \cite{adler1998protocols} or
\cite{patrascu06coding}).

In addition to the natural goal of communication efficiency, a common
requirement for such coding schemes is also to have a {\em computationally
  efficient} decoding procedure. Our goal here is for the decoding
time to be polynomial/linear in $N$ (which is the best we can hope for
without further assumptions --- the input to the decoder is the
distribution $\mu$, of potentially $\Omega(N)$ description
size)\footnote{With further assumptions---e.g., preprocessing---one
  may ask for sublinear runtime, of the order of $\poly(|S|, \lg N)$,
  as was accomplished in some compressed sensing literature; see,
  e.g., \cite{gilbert2010sparse, gilbert2012approximate}.}.

Without further assumptions on the distribution $\sigma$, this problem
does not admit any viable solutions: both communication and
computation are essentially doomed. Indeed, \cite{adler1998protocols,
  patrascu06coding} show that the trivial bound of $\sim N$
communication is required, even when the entropy of $\sigma$ is much
smaller. We note that \cite{adler1998protocols} circumvented this
barrier by allowing two-way communication where the decoder can send
much larger messages back to encoder, whereas we focus on purely {\em
  one-way} protocols only.  As for the distributional setting, a
generic (non-product) prior distribution $\sigma$ has a high
description complexity (exponential in $N$, or max set size), thus
dooming the time-efficiency of any decoding scheme.

In this paper, we consider the most natural class of priors $\sigma$
of i.i.d.~items: the sets $S \sim \mu^k$ are comprised of $k$ items, each drawn
independently from some distribution $\mu$ over $[N]$. We note that
this a common assumption, implicitly assumed in (vanilla) compressed
sensing, as well as classic (symmetric information) source-coding problems.  

For this setting, we develop protocols that achieve efficient decoding
time, and competitive communication costs. Our coding scheme is {\em
  linear}---the encoding is $C\cdot \one_S$ where $C$ is the coding
matrix and $\one_S$ is the indicator vector of the set $S$---which is
a further desirable property of coding scheme. This property is
similar to the one imposed in compressed sensing. Linearity
facilitates quick and simple updates to the message in
streaming/dynamic environments (e.g., in the IoT application above) as
the message can be simply updated as items are added one by one to the
set $S$.

\newcommand{\idlen}{\lg N} 
\newcommand{\numtag}{N}
\newcommand{\hash}[1]{h(#1)} 
\newcommand{\alltags}{[N]}

\subsection{Relation to Problems in Prior Literature}
\label{sec:prior}

Our problem relates to many other problems studied previously, but,
surprisingly, has not been explicitly studied. When there is {\em no
  side information}, the problem is the classic problem of coding a
set $S$.  Without requiring linearity, a trivial solution is to append
the indices of items in $S$, yielding communication $k\lg \numtag$ for
sets $S$ of size $k$.\footnote{We use $\lg$ to denote base-$2$
  logarithm.} If we further require linearity, then the problem
becomes a variant of compressed sensing. A slight caveat is that the
compressed sensing schemes usually work over reals \cite{CT, Don}, and
the vector $C\cdot \one_S$ is a real vector, which raises the issue of
rounding and real number representation. Nevertheless, it is possible
to do compressed sensing over the $\F_2$ field; see, e.g.,
\cite{draper2009compressed, seong2013necessary, li2014robust, DV13}.

Another related model is source coding, where both the encoder and the
decoder have access to some prior distribution $\mu$, and the set $S$
is composed of $k$ items i.i.d.~items drawn from $\mu$. Then a
(near-)optimal solution can be obtained via, say, {\em Huffman coding}
\cite{huffman1952method}. The length of the compression of a set $S$
is $\sum_{i\in S} \lceil\lg 1/\mu(i)\rceil$, which, in expectation, is
upper bounded by $k\cdot H(\mu)+k$, close to the information-theoretic
optimum of $k\cdot H(\mu)$ (up to the rounding issues).

When the side information is not known to the encoder (as it is in our
case), the problem becomes the classic asymmetric transmission problem
\cite{adler1998protocols, laber2002improved, ghazizadeh2001new,
  watkinson2001new, patrascu06coding} (see also
\cite{xiong2004distributed}). In this problem, the encoder generates
an item from a probability distribution $\mu$ and needs to communicate
its identity to the router/server (decoder). The goal is again to reach the
information capacity of $\approx H(\mu)$. While there are protocols
that achieve such capacity, the protocols require {\em two-way}
communication---the backchannel from the decoder to the encoder is on
the order of $\Omega(\lg N)$ bits. Furthermore, this is necessary:
\cite{adler1998protocols} show that either the encoder or decoder has
to communicate the trivial $\Omega(\lg N)$ bits
\cite{adler1998protocols} (see also the follow-up work of
\cite{patrascu06coding} for a lower bound on the number of interactive
rounds required).

In contrast, our protocols use one-way communication only. We
circumvent the above lower bound by exploiting the fact that the
encoder sends a {\em set} $S$ of items, instead of a single one, with
a randomized protocol. In particular, we can amortize the lower bound
of $\Omega(\lg N)$ against $|S|$ items. In other words, in our
setting, we encode a set $S$ using $m\ge \lg N$ bits, with the goal of
achieving $m \ll O(|S|\cdot \lg N)$ where possible.

Finally, we remark that the problem also falls under the umbrella of
model-based compressed sensing, where one generally assumes some prior
knowledge on {\em the possible structure} (model) of the set $S$
(beyond, say, an upper bound on its size); see, e.g.,
\cite{baraniuk2010model}. While the asymmetry is typically not an
explicit goal, the encoding schemes are usually agnostic to this prior
knowledge (e.g., the coding uses the usual matrix with random Gaussian
entries), and hence, in fact, constitute an asymmetric coding scheme.

\subsection{Formal Problem Setup}

There are a few ways to formalize our problem, and hence we introduce
three related definitions below, of growing generality. As before,
there is a universe $[N] := \cbr{1,2,\dotsc,N}$ of items.  For a
given set $S \subseteq [N]$, the encoder $\Enc \colon 2^{[N]} \to
\cbr{0,1}^m$ must construct a (possibly randomized) message $y :=
\Enc(S)$ of at most $m$ bits, where $m$ is the allowed message length,
fixed in advance. The decoder $\Dec_\star \colon \cbr{0,1}^m \to
2^{[N]}$, for some side-information $\star$, must produce a set $\hat
S := \Dec_\star(y)$ from the message $y$ such that $\hat S = S$ with,
say, at least $1-\delta$ probability, where $\delta$ is the error
probability parameter (think $\delta=0.1$).  Note that, when the side
information $\star$ is null, this task is generally impossible unless
$m \geq \lg 2^N = N$. Note that the encoder's message does not depend on
the side information, i.e., the encoding function $\Enc(S)$ is
\emph{oblivious} (in the information theory literature this is
referred to as \emph{universal compression} \cite{CSV03,HU14}; see
also Section \ref{sec:discussion}).

To measure the optimality of a coding scheme, we compare our message
lengths to the information-theoretic minimum, which we denote by the
parameter $m^*$ (which is a function of $\star$). In particular, for
$\alpha\ge1$, a coding scheme is called {\em $\alpha$-competitive} if it
uses $m$ bits while the ``information-theoretic optimal'' is $m^*\ge
m/\alpha$ bits. Note that the value of ``information-theoretic
optimal'' is not obvious, and in fact will differ between different
definition.

There are also a few ways to measure the success of a scheme.  We now
introduce a few related definitions of asymmetric coding in the order of
generality.

Following the discussion from before, one natural way to model the
side information is via a prior distribution $\sigma$ on subsets of
$[N]$. In particular, we assume $\sigma$ is a distribution on $k$
items, each drawn from a distribution $\mu$ on $[N]$.

\begin{definition}
\label{def:entropyAC}
For $N,m,\alpha\ge 1$, a (randomized) scheme $\acs=(\Enc,\Dec)$
is {\em entropy-asymmetric-coding} $\alpha$-competitive scheme if: for any
integer $k$, and prior $\mu$ on $[N]$ such that $k\cdot H(\mu)\le
m/\alpha$, we have the following where the prior $\sigma$ generates a set
of $k$ items drawn iid from $\mu$:
$$
\Pr_{\acs,S\sim \sigma}[\Dec_\sigma(\Enc(S))=S]\ge 1-\delta.
$$

We clarify that the randomness of the encoder and decoder is via a
shared random string, which is an (auxiliary) input to both $\Enc$
and $\Dec$.
\end{definition}

Note that $m^*=k\cdot H(\mu)$ is the lower bound on communication
necessary to transmit a set $S$ of $k$ items drawn iid from $\mu$. The
trivial scheme would achieve a bound\footnote{The more precise bound
  is $\lg {N \choose k}\approx k\lg N/k$, but since we think of $k\ll
  N$, this amounts to a negligible difference.} of $k\lg N$, which can be much
higher than $kH(\mu)$.

We now consider a slightly more general definition, where we do not
need to fix the size $k$ of $S$, but rather be ``adaptive'' to the
number of items in the set $S$, in the analogy to what the Huffman
coding achieves in the symmetric case.

\begin{definition}
For $N,m,\alpha\ge1$, a (randomized) scheme $\acs=(\Enc,\Dec)$ is said to
be a {\em Huffman-asymmetric-coding} $\alpha$-competitive scheme if:
for any distribution $\mu$ over $[N]$, if the set $S$ satisfies
\begin{equation}
\label{eqn:huffmanBound}
\sum_{i\in S} \lg 1/\mu(i) \le m^*,
\end{equation}
where $m^*=m/\alpha$, then
$$
\Pr_{\acs}[\Dec_\mu(\Enc(S))=S]\ge 1-\delta.
$$
\end{definition}

In particular, a Huffman-asymmetric-coding $1$-competitive scheme
matches the performance of the aforementioned Huffman coding (where
the encoder knows the prior $\mu$), for $\delta=0$
(deterministically). We also note that Eqn.~\eqref{eqn:huffmanBound}
(with $\alpha=1$) is the tightest condition we can require in order
for a set $S$ to be decodable with a classic Huffman code. Hence, the
above definition asks to match the efficiency of the Huffman code
(symmetric information setting) in the {\em asymmetric} setting, up to
$\alpha$-factor loss in communication.

It is not hard to note that Huffman-asymmetric-coding scheme is more
general than the entropy-asymmetric-coding scheme: if we pick a random
set $S$ as in Def.~\ref{def:entropyAC}, then it satisfies
Eqn.~\eqref{eqn:huffmanBound} (up to a small loss in communication
efficiency). See Claim~\ref{huffman-to-entropy} in Appendix
\ref{sec_distribution_vs_lists}.

Finally, we give the most general definition, which is the most
natural from an algorithmic perspective, but is less operational than
the two above. It stems from the observation than any desirable
encoding/decoding scheme is (implicitly) specifying a \emph{list}
(ordered set) $L \subseteq 2^{[N]}$ of subsets $S\subseteq[N]$ that
are decoded correctly. It is immediate to see that any such list $L$
can have at most $2^m$ such sets. In the presence of a prior
distribution $\sigma$, one could take these sets to be the ``most
likely'' in $\sigma$ (with ties broken arbitrarily).

\begin{definition}
For $N,m,\alpha\ge1$,
a (randomized) scheme $\acs=(\Enc,\Dec)$ is said to be a {\em
  list-asymmetric-coding} $\alpha$-competitive scheme if: for any list
$L$ of sets $S\subseteq [N]$, where $|L|\le 2^{m/\alpha}$, and any
$S\in L$, we have that:
$$
\Pr_{\acs}[\Dec_L(\Enc(S))=S]\ge 1-\delta.
$$
\end{definition}

Again, the latter definition is more general than both the
definitions. In particular, a list-asymmetric-coding scheme is also a
Huffman-asymmetric-coding scheme: given a prior $\mu$, just fix the
list $L$ to be the sets satisfying condition
\eqref{eqn:huffmanBound}. It is easy to see that the
size of the list will be $\le e2^{m/\alpha}$ (which results in just an
additive $\lg e$ additive loss in communication); see details in
Claim~\ref{list-to-huffman} in Appendix~\ref{sec_distribution_vs_lists}.

The last definition has the major downside that one has to specify a
list $L$ to the decoder, which is exponential in $m$, thus affecting
the computational efficiency of a coding scheme. Therefore, for
algorithmic efficiency, it is more natural to work with the
Huffman-asymmetric-coding definition, which is the focus here.

\subsection{Our Results}

First, we establish that any asymmetric-coding scheme must be randomized
if it is to non-trivially exploit the prior $\mu$ or list $L$. In
particular, if $\delta=0$ (i.e., no randomization), then, there exists
some priors where the optimal communication in the symmetric case is
$m^*= O(|S|\cdot \lg |S|)$, but any asymmetric-coding scheme must have
$m\approx \Theta(|S|\cdot \lg N)$. See details in
Section~\ref{sec:lowerBound}.

Second, as a warm-up, we show a simple scheme that solves the most
general definition, of list-asymmetric-coding scheme, but which is not
computationally efficient.

\begin{theorem}[Information-theoretic; see Section~\ref{sec:random}]
\label{thm:randomCode}
Fix error probability $\delta>0$. There is an $\alpha$-competitive
list-asymmetric-coding scheme with
$\alpha=\tfrac{m}{m-\lg1/\delta}=1+o(1)$, while achieving
error probability of $\delta$.
\end{theorem}

The scheme is a standard one: a random linear code. In particular,
pick a random $C\in M_{m\times N}(\F_2)$, and set $\Enc(S)=C\cdot
\one_S$ (all computations are done in $\F_2$). The decoder $\Dec(y)$
is the ``maximum likelihood'' decoder: for a given list $L$, go over
the list in order and output the first set $\hat S\in L$ such that
$C\one_{\hat S}=y$. See Section \ref{sec:random} for further details
and proofs.

While the above scheme achieves the information-theoretic bound (up to
additive $\lg 1/\delta$), it is {\em
  not computationally-efficient} and requires runtime of about
$\Omega(2^m)$. Even when the list $L$ is somehow more efficiently
represented (e.g., all sets $S$ that satisfy the Huffman condition
Eqn.~\eqref{eqn:huffmanBound}), the problem appears computationally
hard. In particular, it is a variant of the classic problem of
decoding random linear codes. Obtaining a coding scheme with faster
decoding is precisely the focal point of our work:

\begin{quote}
\bf Main goal: \rm \emph{ Develop computationally efficient oblivious
  compression schemes, that have only $\text{poly}(N)$
  encoding/decoding time, at the expense of a (mild, multiplicative)
  overhead in communication cost compared to random codes
  ($\alpha$-competitive).}
\end{quote}

Our main result is the design of a {\em computationally-efficient},
Huffman-asymmetric-coding scheme which is optimal up to a
$O(\log\log N)$-factor loss in the message length.

\begin{theorem}[Main; see Section~\ref{sec:main}]
\label{thm:main}
Fix target message length $m>\lg N+4$, and error probability
$\delta\ge 1/\lg N$. There is a linear
Huffman-asymmetric-coding scheme, which is $O(\log\log N)$-competitive, and has
$\text{poly}(N)$ decoding time and error probability of $\delta$.
\end{theorem}

\subsection{Technical Overview of Theorem \ref{thm:main}}

The proof of Theorem \ref{thm:main} is based on a ``multi-level"
coding scheme.  The basic building block of our ``multi-level" coding
scheme is the \emph{uniform} compressed sensing scheme of \cite{DV13},
which is the finite-alphabet equivalent of standard compressed sensing
schemes (with a ``uniform'' prior).  In particular, their scheme is a
computationally efficient linear sparse recovery scheme for $k$-sparse
vectors in $\F_2^N$, using $O(k\log N)$ bits. Their (deterministic)
scheme relies on \emph{syndrome decoding} of linear codes, which allows
to decode in polynomial time any $k$-sparse vector $x \in \F^N_2$,
using the \emph{parity check} matrix $C_{RS}$ of Reed-Solomon codes
with the appropriate rate/dimension generated by a binary symmetric
(BSC) channel (see Section \ref{sec:syndrome} for details).

Recall that in our setup, the prior $\mu$ is nonuniform and
\emph{unknown} to the encoder. We view the ground set of $[N]$ items
as being partitioned into $T$ buckets of doubly-exponentially decaying
probabilities w.r.t.~$\mu$, where bucket $B_i$ contains all elements
with probability between $2^{2^{-i}}$ and $2^{-2^{i+1}}$
w.r.t.~$\mu$. This allows us to set $T$ to be
\emph{doubly-logarithmic}, i.e., $T=O(\lg \lg N)$.

The encoder sends $T$ concatenated messages, where the goal of the
$i^{th}$ message is to allow the decoder to decode the subset $S \cap
B_i$, where $S\sim \mu^k$ is the input set at the encoder. For each
``level'' $i$, the encoder uses an appropriately-sized sensing matrix
$C^{(i)}_{RS}$, whose dimensions are determined by the (worst-case)
number of elements that could be encoded from $B_i$ (here we
implicitly assume that $\mu$ is uniform on $B_i$, which may lose a
factor of $\le 2$ w.r.t the optimal message size per item, since the
{\em encoding lengths} of items in $B_i$ are within a factor 2). Since
in the $i$th step we only need to distinguish items in $B_i$, the
encoder first \emph{hashes} the set $S$ to the minimal universe
$N_i\ll N$ that still ensures collision-freeness in $B_i$ (using a
\emph{public} hash function shared by the encoder and the decoder),
and $C^{(i)}_{RS}$ is applied to the \emph{hashed} vector in the
reduced universe. This carefully-chosen universe-reduction
``preprocessing" step is essential to save on communication---e.g.,
using \cite{DV13} on $k$ items will cost us only $\sim k\log N_i\ll
k\log N$. Note that, the encoder doesn't actually know the items
$B_i$, and hence we don't know the items $S\cap B_i$ to be encoded in
the level $i$ either.  Instead, the level $i$ encoding will contain
all items $S$ (this is precisely where we lose the $O(\log\log
N)$-factor in communication overall), and the identification of the set
$S\cap B_i$ is done at decoding time only, as described next.

Our decoding procedure is \emph{adaptive} and runs in $T$ successive
steps. In the $i^{th}$ step, we assume we've already successfully
decoded items $S\cap B_{<i}=S\cap (B_1\cup B_2\cup\ldots B_{i-1})$.
The decoder then ``peels off'' the encoding of $S\cap B_{<i}$ from the
original message that it has received. This step crucially uses the
\emph{linearity} of the encoding scheme. The remaining $i^{th}$ level
message now encodes items $S\cap (B_i\cup B_{i+1}\cup\ldots B_{T})$,
which allows us to decode $S\cap B_i$. Note that, in addition to the
aforementioned required property of no collisions inside $B_i$, we
also need universe $[N_i]$ to be sufficiently large so that there are
no collisions between items $B_i$ and in $S\cap B_{>i}$ --- otherwise
we may misidentify an item from $S\cap B_{>i}$ as being in
$B_i$. Luckily, as $|S\cap B_{>i}|\le |S|$ is generally much smaller
than $|B_i|$, this new condition on $N_i$ does not ultimately
influence the communication bound.  Note that, at level $i$, the
decoder will decode any item in $B_i$, and potentially identify that there
exist items $S\cap B_{>i}$ (which will be left for the subsequent steps).

We present the full details of our coding scheme and its analysis in
Section~\ref{sec:main}.

\subsection{Discussion and Open Problems}
\label{sec:discussion}

Finally, on a somewhat different note, noiseless compression in asymmetric scenarios was also
previously studied in the information theory literature, in the
context of \emph{universal compression} (see e.g.,
\cite{CSV03,HU14,DV13} and references therein).  This line of work
exploits an elegant connection between channel coding and
source coding, via \emph{syndrome-decoding},
a connection that also plays an important role as a sub-procedure in our main result (Theorem \ref{thm:main},
see also the discussion in Section \ref{sec:syndrome}).
These works exhibit (fixed-length) codes with efficient encoding and decoding procedures
against a subclass of discrete \emph{memoryless} channels (DMCs),
e.g., via belief-propagation for LDPC codes \cite{CSV03} and Turbo codes \cite{GfZ02}.
The main difference of our model is that the aforementioned line of work
relies on an interpretation of the set to be encoded ($S$) as a (sparse) additive noise vector
generated by a discrete \emph{memoryless} channel (or even
further restricted symmetric channels such as BSC), where each
coordinate in $[N]$ is corrupted by the channel \emph{independently with
identical} probability.  Indeed, decoding
procedures such as belief-propagation algorithms are only
guaranteed to converge under specific DMC channels such as BSC.
This assumption is equivalent in our model to considering only
\emph{i.i.d}
distributions $\mu$ on the $[N]$ coordinates (i.e., each item $i$ is
present $i\in S$ iid with certain probability), whereas we wish to
deal with \emph{arbitrary} product distributions $\mu^k$, $\mu \in
\Delta([N])$ (where $\Delta([N])$ denotes the set of all distributions over $[N]$).

\paragraph{Open questions.} As we view this work as an initial step in
the study of asymmetric compression, there are a few natural aspects
of our assumptions that require further research:

\begin{itemize}[leftmargin=2em]
\item
The most straightforward open question is whether the message length
for product distributions over subsets of $[N]$ can be improved from
$O_\delta(\lg\lg N)$ multiplicative overhead to $O(\lg(1/\delta))$
overhead, or even further to $O(\lg(1/\delta))$ {\em additive}
overhead (matching the information bound of the baseline scheme from
Theorem~\ref{thm:randomCode}), while insisting on $\text{poly}(N)$
decoding time. We note that even the scheme of \cite{DV13} (for the
uniform prior case) is only 2-competitive.
\item
As hinted before, we may also want decoding time which is sublinear in
$N$, e.g., $\text{poly}(m,\log N)$. Note that this may be possible
only if we allow the decoder to do preprocessing---otherwise, already
its input $\mu$ has $\Omega(N)$ description size.
\item
Are the above goals simpler if we allow {\em non-linear} coding? Our
scheme is linear, and we do not know if there exist more efficient
non-linear coding schemes.

\item
Another important direction is to identify other natural instances of
\emph{non-product} distributions $\sigma$, where the problem is
meaningful and poly-time, competitive coding schemes exist. As
mentioned before, such a distribution $\sigma$ must at minimum have a
\emph{succinct} description.  A natural candidate family for modeling
such succinct joint distributions on subsets of $[N]$ are
\emph{graphical models} \cite{WJ08}. It would be very interesting to
develop 
compete with the (possibly much lower) entropy benchmark of joint
distributions generated by low-order graphical models.

\item
Finally, one may want to construct schemes that have a somewhat better
probability guarantee (somewhat akin to ``for all'' vs ``for each''
guarantee). While fully deterministic schemes are impossible, it may
be possible to obtain the following guarantee: with probability
$1-\delta$, the decoder decodes correctly {\em any} set $S\in L$. It
turns out that this is possible for the random code solution (see
Corollary~\ref{cor:mldec-forall}). It would be interesting if our main
(computationally-efficient) result can be extended to this case as
well.
\end{itemize}

\section{A Basic Scheme: Random Linear Codes}
\label{sec:random}

We establish Theorem \ref{thm:randomCode} by designing a
list-asymmetric-coding scheme via a \emph{random linear} code. It
achieves essentially optimal communication (up to additive $O(1)$
bits), nearly matching the performance of the symmetric-information
schemes. The runtime of this scheme is exponential in $m$.

Consider a randomized linear scheme where $C$ is a uniformly random
matrix $C \in \bbF_2^{m \times N}$, and $\Enc(S)=C\cdot \one_S$. The
decoder for a list $L = (S_1, S_2, \dotsc, S_{\abs{L }})$ is the
``maximum likelihood'' decoder: given the message $y$, the decoder
returns the \emph{first} set $S$ in the list $L $ such that $\Enc(S) =
y$:
\begin{equation*}
  \MLDec_L (y)
  \ := \
  S_{\min\cbr{ t \in \sbr{\abs{L }} : \Enc(S_t) = y }}
  \,.
\end{equation*}
(The random matrix $C$ is determined using the public random bits).
For brevity, we call this the \emph{random linear scheme}.

The next lemma
establishes that the random linear scheme is a list-asymmetric-coding
scheme for any $\delta \in \intoo{0,1}$ and any list of at most
$2^{m}\cdot \delta=2^{m-\lg1/\delta}$ subsets of $[N]$. It implies
Theorem~\ref{thm:randomCode} since the competitiveness is
$\alpha=\tfrac{m}{m-\lg 1/\delta}$.

\begin{lemma}
  \label{thm:mldec}
  Let $C$ be a random $m\times N$ binary matrix. Then for any list $L$
  of $|L|\le 2^{m}$ subsets of $[N]$, and any $S \in L$:
  \begin{equation*}
    \Pr_C\del{
      \MLDec_L (C\cdot \iS) = S
    }
    \ \geq \
    1 - \del{|L| - 1} 2^{-m}
    \,.
  \end{equation*}
\end{lemma}

\begin{proof}
For any pair of sets $S, S'$ in the list $L $, we use $S \prec_L S'$ to
denote that $S$ appears before $S'$ in $L $.
We also let $S \triangle S' := (S \setminus S') \cup (S' \setminus S)$ denote
the symmetric difference between $S$ and $S'$.
Finally, for $i \in [N]$ and $j \in [m]$, we let $c_i(j)$ denote the $j$-th
entry of the code word $c_i$.

  The decoder outputs a set $\hat S := \MLDec_L (\Enc(S)) \neq S$ if and only
  if there is exists $S' \neq S$ such that $S' \prec_L S$ and $\sum_{i \in
  S'} c_i = \sum_{i \in S} c_i$.
  For any set $S' \prec_L S$ in $L $,
  \begin{align*}
    \Pr\del{
      \sum_{i \in S'} c_i
      =
      \sum_{i \in S} c_i
    }
    & \ = \
    \prod_{j=1}^m
    \Pr\del{
      \sum_{i \in S'} c_i(j)
      =
      \sum_{i \in S} c_i(j)
    }
    \\
    & \ = \
    \prod_{j=1}^m
    \Pr\del{
      \sum_{i \in S' \triangle S} c_i(j)
      =
      0
    }
    \ = \
    2^{-m}
    \,.
  \end{align*}
  By a union bound,
  \begin{align*}
    \Pr\del{
      \MLDec_L (\Enc(S)) \neq S
    }
    & \ = \
    \Pr\del{
      \exists S' \prec_L S \centerdot
      \sum_{i \in S'} c_i
      =
      \sum_{i \in S} c_i
    }
    \\
    & \ \leq \
    \sum_{S' \prec_L S}
    \Pr\del{
      \sum_{i \in S'} c_i
      =
      \sum_{i \in S} c_i
    }
    \\
    & \ \leq \
    \del{\abs{L } - 1} 2^{-m}
    \,.
    \qedhere
  \end{align*}
\end{proof}


In fact, one can prove a slightly stronger guarantee of success: that,
for any fixed list $L$, with probability at least $1-\delta$, the
decoder decodes correctly {\em any} set $S\in L$. This leads to
slightly worse competitiveness: $\alpha=2+o(1)$. In particular,
$m$-sized code can decode only lists of size $2^{m^*}$ where
$m^*=\tfrac{1}{2}(m-\lg 1/\delta)$. The following corollary is
immediate from the above.

\begin{corollary}
  \label{cor:mldec-forall}
  Let $C$ be a random $m\times N$ $0/1$ matrix.
  Then for any list $L $ of subsets of $[N]$,
  \begin{equation*}
    \Pr_C\del{
      \forall S \in L \centerdot
      \MLDec_L (C\cdot \iS) = S
    }
    \ \geq \
    1 - |L|\cdot \del{|L| - 1} 2^{-m}
    \,.
  \end{equation*}
\end{corollary}

\section{Main Result: $O(\log\log N)$-competitive Coding Scheme}
\label{sec:main}

In this section, we prove Theorem~\ref{thm:main}, by designing a
computationally efficient Huffman-asymmetric-coding scheme. The
resulting algorithm is termed the \emph{multi-level scheme} (for
reason that will soon be apparent).

Let $ \Delta([N])$ be the space of all distributions with support $[N]$. Our algorithm supports distributions $\mu$ from the following class
\begin{equation*}
  \cM
  \ := \
  \cbr{
    \mu \in \Delta([N]) : 1/4N \leq \mu(i) < 1/2, \ \forall i \in [N]
  }
  \,.
\end{equation*}
While this is a restriction from a general distribution
$\mu\in\Delta([N])$, it is without loss of generality: we can
transform any distribution into a distribution $\mu''\in\cM$ (up to a
loss of at most factor 2 in the communication bound). First, if there
are items $i^*$ with probability more than 1/3, make them with
probability $1/3$: set $\mu'(i^*)=1/3$.  Second, all the probabilities
that are too small can be brought up to at least $1/4N$, while
affecting the other probabilities only by a constant as follows: (1)
construct $\mu'(i)=\max\{\mu(i), 1/2N\}$ (except for items $i^*$), (2)
let $\zeta=\sum_i \mu'(i)\le \sum_i (\mu(i)+1/2N)=1.5$, and (3) set
$\mu''(i)=\tfrac{1}{\zeta}\mu'(i)$. It's not hard to verify now that
$\mu''\in \cM$, as well as that $\mu''(i^*)\le 1/2$ and for the other
items $\lg1/\mu''(i)\le 2\lg 1/\mu(i)$.  We also assume that $m\ge \lg
N+4$.

Our scheme $\acs=(\Enc,\Dec)$ uses $T := \lg\lg(4N)$ levels, each
parametrized by positive integers $D_t, m_t$ to be determined
later. We use uniformly random hash functions
\begin{equation*}
  h_t \colon [N] \to [D_t]
  \,
\end{equation*}
where the hash functions are determined using shared public randomness.
The scheme also uses a family of $T$ (deterministic) linear codes, 
$C^{(t)} = \sbr[0]{
  \begin{smallmatrix} c^{(t)}_1 & c^{(t)}_2 & \dotso & c^{(t)}_{D_t}
  \end{smallmatrix}
} \in \bbF_N^{m_t \times D_t}$ for $t \in [T]$, which are specified in the next subsection.
Each matrix $C^{(t)}$ shall be designed to support efficient decoding of \emph{every} $\left(\frac{m_t}{2\lg D_t}\right)$-sparse
vector. We now turn to the formal construction.


\subsection{One level: sensing matrices $C^{(t)}$} \label{sec:syndrome}

For each level of our scheme, the basic building block is the
compressed-sensing matrices designed in the work of \cite{DV13}. These
deterministic constructions produce $m \times N$ linear codes
(matrices over some finite field) that can decode \emph{any}
$k$-sparse vector $x\in \bbF_2^N$ (i.e., any subset of size at most
$k$), where $k := m/(2\lg N)$, in time \emph{polynomial} in $m$ and
$N$.  Note that such a compression scheme is essentially optimal --
the number of $k$-sparse subsets in $[N]$ is ${N \choose k} \approx
2^{k\lg(N/k)}$, hence any deterministic encoding scheme for this
problem must use at least $k\lg(N/k) \approx m$ bits of communication.

We now state the formal theorem from \cite{DV13}.  The theorem relies
on an elegant connection between channel coding and source coding (via
``syndrome decoding''). The central object is the \emph{parity check}
matrix of a \emph{Reed-Solomon} code (see e.g., \cite{Roth06}). To
this end, we denote by $[N,r, d]_{q}$ a Reed-Solomon code over the
alphabet $\bbF_q$ ($q\geq \lg N$), whose codeword length is $N$,
number of codewords is $q^r$, and the minimum Hamming distance between
codewords is $d$ (i.e., the code can correct up to $(d-1)/2$ errors).
Our multi-level scheme uses the following theorem in a black-box
fashion.

\begin{theorem}[Efficient deterministic compressed sensing, \cite{DV13}] \label{thm_DV13}
Let $\bP^N_k  \in \bbF_N^{m\times N}$
be the parity-check matrix of a $[N,N-2k, 2k+1]_{\bbF_N}$ Reed-Solomon code\footnote{We assume here that $N$ is a power
of 2. Otherwise,  replace it with $N' := 2^{\lceil \lg N\rceil}$.}, where $m=2k\lceil \lg N\rceil$.
There is a (deterministic) decoding algorithm that recovers any $k$-sparse vector in $\bbF_2^N$ (i.e., $x\in  \binom{[N]}{k}$)
from $\bP^N_k  \cdot x$ using $O(Nk\lg^2N)$ operations over $\bbF_2$.
In particular, $\bP^N_k \cdot x$ uniquely determines $x$ using $m= 2k\lceil \lg N\rceil$ linear measurements.
\end{theorem}

The rough idea behind this result (which was used in the past) is to
think of $k$-sparse vectors in $\bbF_2^N$ as a sparse \emph{noise}
vector introduced by a discrete memoryless channel, and then use the
efficient \emph{syndrome-decoding} algorithm for Reed-Solomon codes of
Berlekamp and Massey (see \cite{Roth06}) which recovers the noise
vector (i.e., our desired $k$-sparse subset) from the parity check
matrix $\bP^N_k$.

Of course, the main difference from the setup of Theorem \ref{thm_DV13} and our setup, is that in our case
the original distribution on subsets (i.e., sparse vectors) may be very far from uniform. Nonetheless, our multi-level
scheme uses the construction of \cite{DV13} in each layer. More precisely, for level $t$ of our scheme, our scheme
shall set the matrix $C^{(t)}$ to be the parity-check matrix $\bP^N_k$
with parameters $N:= D_t$, $k := m_t/(2\lg D_t)$ (i.e., it is a matrix
of size $m_t\times D_t$).
This will become clearer in the next section where we present the entire multi-level scheme.


\subsection{Description and Analysis of the Multi-level Scheme}

As mentioned in the previous section, the encoding and decoding of the
input ($S\subseteq [N]$) is defined by an iterative procedure
consisting of $T$ levels, and crucially relies on the linearity of the
encoding in each level. Let $\{D_t\}_{t \in [T]}$ and $\{m_t\}_{t\in
  [T]}$ be numbers to be determined later. The encoder is described in
\Cref{alg:multilevel-enc}, and the decoder is described in
\Cref{alg:multilevel-dec}.

\begin{figure}[h!]
  \renewcommand\algorithmicrequire{\textbf{input}}
  \renewcommand\algorithmicensure{\textbf{output}}
  \algcaption{$\Enc$ for multi-level scheme}
  \label{alg:multilevel-enc}
  \begin{algorithmic}[1]
    \REQUIRE
      subset $S \subseteq [N]$ \; (represented as the indicator vector $\one_S \in \{0,1\}^N$).

    \ENSURE
      message $y \in \cbr{0,1}^m$.

      For each $t\in [T]$, let $y^{(t)} := \sum_{i\in S} C^{(t)}\cdot \one_{\{h_t(i)\}}$, where $C^{(t)}$ is the $m_t\times D_t$
      matrix $\bP^{N_t}_{k_t}$ \\
      from Theorem \ref{thm_DV13}, instantiated with $N_t:= D_t$, $k_t := m_t/(2\lg D_t)$. i.e., $y^{(t)} = \sum_{i \in S} c^{(t)}_{h_t(i)}$.
      \RETURN
	concatenated string $y := (y^{(1)},y^{(2)},\dotsc,y^{(T)})$ 

  \end{algorithmic}
  \bottomrule

  \bigskip

  \algcaption{$\Dec_\mu$ for multi-level scheme}
  \label{alg:multilevel-dec}
  \begin{algorithmic}[1]
    \REQUIRE
      message $y = (y^{(1)},y^{(2)},\dotsc,y^{(T)}) \in \cbr{0,1}^m$, and a prior distribution $\mu\in \cM_m$.

    \ENSURE
      subset $\hat S \subseteq [N]$.

    \STATE
      Let $B_t := \cbr[0]{ i \in [N] : 2^{-2^t} \leq \mu(i) < 2^{-2^{t-1}} }$
      for $t \in [T]$.

    \STATE
      Initialize $\hat S := \emptyset$.

    \FOR{$t=1,2,\dotsc,T$}

      \STATE
        Let $\hat z^{(t)}$ be the output of the decoder for $C^{(t)}$ applied to
        $y^{(t)}$, guaranteed by Theorem \ref{thm_DV13}.

      \FOR{each $i \in B_t$}

        \IF{$\hat z^{(t)}_{h_t(i)} = 1$}

          \STATE
            Let $\hat S := \hat S \cup \cbr{i}$.

          \FOR{$\tau = t+1,t+2,\dotsc,T$}

            \STATE
              Let $y^{(\tau)} := y^{(\tau)} - c^{(\tau)}_{h_\tau(i)}$.

          \ENDFOR

        \ENDIF

      \ENDFOR

    \ENDFOR

    \RETURN $\hat S$

  \end{algorithmic}

  \bottomrule
\end{figure}


We now turn to the analysis of the scheme, whose centerpiece is the following theorem.

\newcommand\ceil[1]{\left\lceil #1 \right\rceil}
\begin{theorem}
  \label{thm:multilevel}
  Fix $\delta \in \intoo{0,1}$ and positive integer $m^*$.
  Set
  \begin{equation}
    D_t \ := \
    \ceil{
      \frac{T}{\delta}
      \cdot
      \del{
        2^{2\cdot 2^{t}}/2 + \frac{(m^*)^2}{2^{2t}}
      }
    }
    \,, \quad t \in [T] \, ,
    \label{eq:cond-D_t}
  \end{equation}
  and
  \begin{equation}
    m_t \ := \
    \ceil{
     2\lg D_t\cdot \min\left\{\frac{m^*}{2^{t-1}}, \frac{4m^*}{\lg m^*}\right\}
    }
    \,, \quad t \in [T] .
    \label{eq:cond-m_t}
  \end{equation}
  Then for any $\mu \in \cM$ and $S$ satisfying
  Eqn.~\eqref{eqn:huffmanBound} with the fixed value of $m^*$, the
  Algorithm~\ref{alg:multilevel-dec} outputs the set $\hat S=\Dec_\mu(\Enc(S))$ satisfying:
  \begin{equation*}
    \Pr[\hat S = S]
    \ \geq \
    1 - \delta
    \,.
  \end{equation*}
\end{theorem}

  We now briefly verify that Theorem~\ref{thm:multilevel} implies
  Theorem~\ref{thm:main}, when we set $m^*=m/\alpha$ where $\alpha=O(\lg\lg
  N+\lg 1/\delta)$. Since $\lg D_t\le \lg 2T/\delta+O(2^t)+O(\lg
  m^*)$, we have $m_t\le O(m^*(1+2^{-t+1}\lg 2T/\delta))$. The total message
  length over all the $T$ levels is thus
  \[
    \sum_{t=1}^T m_t = O(m^*\cdot T)+O(m^*\cdot \lg 2T/\delta)\le
    m^*\cdot \alpha=m.
  \]
  Using Theorem~\ref{thm_DV13}, it is also clear that the running times of 
  Algorithm~\ref{alg:multilevel-enc} and Algorithm~\ref{alg:multilevel-dec} are $\text{poly}(N)$.

\begin{proof}[Proof of \Cref{thm:multilevel}]
  Fix $\mu \in \cM$ and $S$ satisfying Eqn.~\eqref{eqn:huffmanBound}.
  Because every $i \in S$ satisfies $\lg(1/\mu(i)) \leq \lg(4N)$, we may partition
  $S$ into $S_t := S \cap B_t$ for $t \in [T]$.
  Also let $S_{t:T} := S_t \cup S_{t+1} \cup \dotsb \cup S_T$ for $t
  \in [T]$.
  Let $E_t$ be the event in which the following hold:
  \begin{enumerate}[leftmargin=2em]
    \item
      $h_t(i) \neq h_t(j)$ for all distinct $i,j \in B_t$;

    \item
      $h_t(i) \neq h_t(j)$ for all $i \in S_t$ and $j \in S_{t+1:T}$.

  \end{enumerate}
  By definition, every $i \in B_t$ satisfies $\mu(i) \geq 2^{-2^t}$,
  and hence $\abs{B_t} \leq 2^{2^t}$.  
Furthermore, every $i \in  S_{t:T}$ satisfies $\mu(i) \leq 2^{-2^{t-1}}$, 
or equivalently,  $1 \leq \frac{\lg (1/\mu(i))}{2^{t-1}}$. Therefore, it holds that
  \begin{equation} \label{eqn:|S_{t:T}|}
    \abs{S_{t:T}} \leq     \sum_{i \in S_{t:T}} 1  
    \ \leq \
    \sum_{i \in S_{t:T}}
    \frac{
      \lg(1/\mu(i))
    }{
      2^{t-1}
    }
    \ \leq \
    \frac{
      \sum_{i \in S}\lg(1/\mu(i))
    }{
      2^{t-1}
    }
    \ \leq \
    \frac{m^*}{2^{t-1}}
    \,,
  \end{equation}
where the final inequality follows since the set $S$ satisfies Eqn.~\eqref{eqn:huffmanBound}.

Now we note that
  \begin{equation*}
    \abs{S_t} \cdot \abs{S_{t+1:T}}
    \ \leq \
    \frac14 \cdot \abs{S_{t:T}}^2
    \ \leq \
    \frac{(m^*)^2}{2^{2t}}
    \,.
  \end{equation*}
  Therefore, by a union bound, the probability that $E_t$ holds is
  \begin{equation*}
    \Pr(E_t)
    \ \geq \
    1 -
    \del{
      \binom{\abs{B_t}}{2}
      + \abs{S_t} \cdot \abs{S_{t+1:T}}
    } \cdot \frac1{D_t}
    \ \geq \
    1 - \frac{\delta}{T}
    \,,
  \end{equation*}
  where the second inequality uses the choice of $D_t$ in
  Eqn.~\eqref{eq:cond-D_t}.
  By another union bound over all $t \in [T]$, it follows that the event $E :=
  E_1 \cap E_2 \cap \dotsb \cap E_T$ holds with probability at least $1-\delta$.

  For the rest of the analysis, we condition on the occurrence of the event $E$.
  Let $\hat S_t$ be the set of items that \Cref{alg:multilevel-dec} adds to
  $\hat S$ in iteration $t$.
  It suffices to prove that if $y$ is the encoding of items belonging only to buckets $B_t,
  B_{t+1}, \dotsc, B_T$ (i.e., of the indicator vector $\one_{S_{t:T}}$),  
  then upon reaching iteration $t$ of the decoding algorithm, we have
  $\hat S_t = S_t$ (i.e., we argue that in level $t$ we decode precisely the elements in $S_t$).
  Maintaining this invariant is indeed sufficient, because at the end of iteration $t$,
  \Cref{alg:multilevel-dec} subtracts the $C^{(\tau)}$-encoding of elements in
  $\hat S_t \cap B_t$ from $y^{(\tau)}$ for all $\tau > t$. Thus, if $\hat S_t = S_t$,
  then after iteration $t$, the \emph{linearity} of the code implies that the message $y$
  (at least the parts relevant to rounds $>t$) no longer contains the items in
  $S_t$ (and hence $B_t$).

  Since we conditioned on the event $E$, the hash function $h_t$ has no collisions between pairs of items
  in $B_t$, and moreover it has no collisions between items in $S_t$ and items
  in $S \setminus S_t = S_{t+1:T}$ (where we use the assumption that $S =
  S_{t:T}$).
  Therefore, the items in $S_t$ are in one-to-one correspondence with some
  subset of $\supp(z^{(t)})$, where
  \begin{equation*}
    z^{(t)}
    \ := \
    \sum_{i \in S} e_{h_t(i)}
    \,.
  \end{equation*}
  The vector $z^{(t)}$ may have other non-zero entries not in the one-to-one
  correspondence with $S_t$, but they are not the image of any $i \in B_t$ under
  $h_t$.
  This implies that if $\hat z^{(t)} = z^{(t)}$, then $\hat S_t = S_t$.

  We now argue that, indeed, we have $\hat z^{(t)} = z^{(t)}$.
  As argued above, we may assume that $S = S_{t:T}$.
  Observe that $y^{(t)}$ is
  the encoding of $z^{(t)}$ under $C^{(t)}$, i.e., $y^{(t)} = C^{(t)} z^{(t)}$.
  Furthermore, observe that $z^{(t)}$ has at most
  \[ \abs{S_{t:T}} \leq \min\left\{\frac{m^*}{2^{t-1}},\frac{4m^*}{\lg m^*}\right\} \]
  non-entries in total. The first argument in the $\min$ comes from Eqn.~\eqref{eqn:|S_{t:T}|}.
 The second argument in the $\min$ is due to a basic entropic inequality:
at least half of the set $S$ is composed
of items of probability mass at most $2/|S|$, and thus, by
Eqn.~\eqref{eqn:huffmanBound}, $\tfrac{|S|}{2}\lg\tfrac{|S|}{2} \le m^*$; this in turn implies $|S| \leq 4m^*/\lg m^*$.
  Due to the choice of $m_t$ from Eqn.~\eqref{eq:cond-m_t} and Theorem~\ref{thm_DV13}, the decoding of
  $y^{(t)}$ returns $\hat z^{(t)} = z^{(t)}$ as required.
\end{proof}

\section{Lower Bound for Deterministic Schemes}
\label{sec:lowerBound}

We show that asymmetric coding schemes need to be randomized in order
to gain advantage from using the side information.  In particular we
show that if the class of priors is sufficiently rich, then no
\emph{deterministic} asymmetric coding scheme can improve over the
trivial baseline communication, even if we allow arbitrary
(non-linear) schemes and arbitrary decoding time. Note that this
separates the asymmetric information case from the symmetric side
information case---since the Huffman code is a deterministic
(near)-optimal algorithm for the symmetric case.

We will prove the lower bound for the entropy-asymmetric-coding case
(the weakest definition).  We consider the family $\cM_{N,k}$ of prior
distributions that consists of all (product) distributions $\mu^k$
where $\mu$ is supported on some subset $M\subset [N]$ of cardinality
$|M|=2k$ (i.e., each $\mu$ defines a list $L=L(\mu)$ of all
$\binom{2k}{k}$ subsets of $[M]$). More formally,
\[ \cM_{N,k} := \left\{ \mu^k  \; | \; \supp(\mu) \subset M, \;\; M\subset[N], |M|=2k  \right\} . \]
Note that for any prior $\mu^k\in \cM_{N,k}$, we have the
information-theoretic minimum communication to be
$m^*=H(\mu^k)=kH(\mu)\leq k\lg(2k)$. However, the following claim
asserts that any deterministic scheme for $S \in \cM_{N,k}$ must spend
essentially the trivial communication of $\Omega(\lg {N\choose
  k})=\Omega(k\lg N/k)$.

\begin{claim}[Deterministic oblivious compression is impossible] \label{thm_det_impossibility}
Any entropy-asymmetric-coding scheme that handles priors
$\sigma=\mu^k\in\cM_{N,k}$, and achieves $\delta=0$, must have
$m=\Omega(k\lg(N/k))$ bits of communication even though the
information-theoretic minimum is $m^*\le k\lg
2k$. This remains true even without requiring linearity or
computational efficiency.
\end{claim}

\begin{proof}
The idea is to use the fact that the encoder is oblivious to $\mu$ in
order to argue that any deterministic encoding scheme can in fact be
used to reconstruct \emph{any} $k$-sparse vector in $\bbF_2^N$ (i.e.,
any subset $S\in \binom{[N]}{k}$).  Clearly, the latter compression
problem requires $\lg\binom{N}{k}$ bits of communication, hence the
claim would follow.
Indeed, we claim that a deterministic scheme $\acs=(\Enc,\Dec)$ that
solves the entropy-asymmetric-coding problem,
must satisfy
\[ \forall \;\; S_1\neq S_2 \subset \binom{[N]}{k} \;\; , \;\; \Enc(S_1)\neq \Enc(S_2) . \]
Indeed, suppose this is false, then there is a pair of subsets $S_1\neq S_2 \subset \binom{[N]}{k}$ which are mapped
by $\acs$ to the same message $$\Enc(S_1) = \Enc(S_2) := \pi .$$
Now, consider the set $M := S_1\cup S_2$ and let $\mu_M$ be the uniform distribution over $M$.
Note that $|M| = |S_1\cup S_2| \leq 2k$, and without loss of generality, assume that $|M| = 2k$ (otherwise, add arbitrary elements of $[N]$ to $M$).
In this case, observe that $\mu_M^k \in \cM_{N,k}$, and that $\Pr_{\mu_M^k}[S_1] = \Pr_{\mu_M^k}[S_2] = 1/|M|^k$.
Therefore, with probability at least
$\delta := 1/(2\cdot|M|^k) = 1/(2\cdot (2k)^k) > 0$, the decoding will fail, since
\begin{multline*}
  \Pr_{S\sim \mu_M^k}\del{\Dec_{\mu_M^k}(\Enc(S)) = S} \\
  \leq 1 - 2\delta\cdot \min\left\{\Pr\del{\Dec_{\mu_M^k}(\pi) = S_1}, \Pr\del{\Dec_{\mu_M^k}(\pi) = S_2}\right\} \leq 1-\delta < 1.
\end{multline*}
But this contradicts the premise that $\acs$ is a deterministic communication scheme with respect to $\cM_{N,k}$.
This proves that the worst-case communication length of any deterministic scheme must be $\Omega(k\lg(N/k))$ bits
even under the class of product distributions.
\end{proof}

\begin{remark}
If arbitrary (non-product) distributions are allowed, it is not hard to turn the above argument into an
\emph{average case} lower bound, for example, by considering the distribution $\sig$
 that chooses $S_1$ or $S_2$ each with probability $1/2$, where $S_1,S_2$ are the ``colliding" sets from above
(note that while $\sig \notin \cM_{N,k}$, $|L(\sig)|=2$).
We also remark that this claim essentially states that prior-oblivious deterministic compression cannot perform any better than standard
(``prior-free") compressed-sensing schemes for $k$-sparse vectors in $\bbF_2^N$, which indeed requires  $\Theta(k\lg(N/k))$ bits/measurements.
\end{remark}



\bibliographystyle{alpha}
\bibliography{refs,andoni,dan,main}

\appendix

\section{Connections Between Different Notions of Asymmetric-coding Schemes}
\label{sec_distribution_vs_lists}

In this section, we show connections between different asymmetric
coding schemes. First we show that a list-asymmetric-coding scheme
implies a Huffman-asymmetric-coding scheme.

\begin{claim}\label{list-to-huffman}
  If $\acs$ is a list-asymmetric-coding scheme with parameters
  $m_l^*$ and $\delta$, then $\acs$ is a Huffman-asymmetric-coding scheme
  with parameters $ m^\star_h \le m^\star_l-\lg e $ and $\delta$, and the same, fixed
  communication bound $m$.
\end{claim}

\begin{proof}
Consider any distribution $\mu$ over $[N]$. Let $L$ be the list of
subsets $S \subseteq [N]$ that satisfy Eqn.~\eqref{eqn:huffmanBound}. We just need to show that the size of $L$ is
less than $e2^{m_h^\star}\le 2^{m_l^\star}$. A set $S$ satisfies Eqn.~\eqref{eqn:huffmanBound} if and only if
$$
\prod_{i\in S} \mu(i) \geq 2^{-m_h^\star}.
$$
On the other hand
\begin{eqnarray*}
\sum_{S \in L} \prod_{i \in S} \mu(i) &\leq& \sum_{S \subseteq [N]} \prod_{i \in S} \mu(i)\\
&=& \sum _{(x_1,\cdots,x_N)\in \{0,1\}^N} \prod_{i=1}^N \mu(i)^{x_i}\\
&=& \sum _{x_1\in \{0,1\}}\mu(1)^{x_1}\sum _{x_2\in \{0,1\}}\mu(2)^{x_2}\cdots \sum _{x_N\in \{0,1\}}\mu(N)^{x_N}\\
&=& (1+\mu(1))(1+\mu(2))\cdots (1+\mu(N))\\
&\leq& e^{\mu(1)}e^{\mu(2)}\cdots e^{\mu(N)}\\
&=& e.
\end{eqnarray*}
Hence the size of list $L$ is less than $e2^{m_h^\star} \leq 2^{m^\star_l}$ and a list-asymmetric-coding scheme for list $L$, with parameters  $m_l^*$ and $\delta$, yields an error probability $\delta$.
\end{proof}

We now show that entropy-asymmetric-coding is the weakest of the three
definitions, in that a list- or Huffman-asymmetric-coding scheme
implies an entropy-asymmetric-coding scheme (with slightly weaker
parameters).  We first define, for any $\delta >0$ and distribution
$\sig \in \Delta(2^{[N]})$, the $\delta$-\emph{approximate cover size}
of $\sig$ as
\[ \cC(\sig,\delta) :=   \min_{m \in \mathbb{N}}  \left\{ \exists L\subseteq \supp(\sig), |L|\leq 2^m \; , \; \sig(L)\geq 1-\delta  \right\}  .  \]

The following claim asserts an upper bound on the cover number in terms of the Shannon entropy of $\sig$.
\begin{claim}[Cover-size vs. Entropy] \label{prop_cover_size_vs_entropy}
For every distribution $\sig$ and $\delta > 0$, it holds that
$$ \Cmu \; \leq \; H(\sig)/\delta.$$
\end{claim}

We remark that the bound is essentially tight, as demonstrated by the
distribution $\sig$ which has an ``atom" of measure $\delta$ and
otherwise uniform on the entire domain.

\begin{proof}
Let $\cG_\delta := \{x :   \lg(1/\sig(x))  \leq H(\sig)/\delta \} $ be the set of elements with ``large" mass under $\sig$.
Indeed, note that $\forall x \in \cG_\delta$ we have $\sig(x) \geq 2^{-H(\sig)/\delta}$, thus it holds that $|\cG_\delta| \leq 2^{H(\sig)/\delta}$.
In order to conclude that $\Cmu \leq H(\sig)/\delta$, it remains to show that $\sig(\cG_\delta) \geq 1-\delta$.
Indeed, Markov's inequality implies that
\[ \sig(\cG_\delta) = 1-\sig(\overline{\cG_\delta}) =  1-  \Pr_{x\sim \sig} \left( \lg \frac{1}{\sig(x)} > \frac{H(\sig)}{\delta} \right)  =
1 - \Pr_{x\sim \sig} \left( \lg \frac{1}{\sig(x)} > \frac{\E\left[ \lg \frac{1}{\sig(x)}\right]}{\delta} \right)  \geq 1- \delta . \]
\end{proof}

The following is a corollary of Claim~\ref{prop_cover_size_vs_entropy}.
\begin{claim}
  If $\acs$ is a list-asymmetric-coding scheme with parameters $m_l^*$
  and $\delta_l$, then $\acs$ can be converted into an
  entropy-asymmetric-coding scheme with parameters
    $m_e^* := \delta_l m_l^*$ and $\delta_e := 2\delta_l$ (and same, fixed
  communication bound $m$).
\end{claim}
\begin{proof}
For any prior $\sigma$ on subsets of $[N]$, there is a list
$L=L(\sig)$ of size at most $2^{H(\sig)/\delta_l}$ which is
``responsible" to $1-\delta_l$ mass of the distribution.\footnote{As
  mentioned before, this ``truncation" of the tail of $\sig$ seems
  inherent to oblivious schemes, as they are \emph{fixed-length}
  encodings.}
So, when the encoding length is
fixed to $m$, Claim \ref{prop_cover_size_vs_entropy} guarantees that
  decoding (w.p.~$1-\delta_l$) all subsets with $\sig(S) \geq 2^{-m_l^*}$ is equivalent
to decoding (w.p.~$1-\delta_l$) all distributions with
Shannon entropy at most $\delta_l m_l^*$.
\end{proof}
Note that $\delta_l m_l^*$ bits are
needed even in the standard compression setup when both parties know
the distribution, hence this notion of decoding is competitive even
with the Shannon entropy benchmark, which is the strongest possible.

Similarly, we can show that a Huffman-asymmetric-coding scheme implies an entropy-asymmetric-coding scheme (with some loss in the communication efficiency).

\begin{claim}\label{huffman-to-entropy}
  If $\acs$ is a Huffman-asymmetric-coding scheme with parameters
  $m_h^*$ and $\delta_h$, then for any $\epsilon \in (0,1)$, $\acs$
  is an entropy-asymmetric-coding scheme with parameters
  \[
    m_e^*
    :=
    \left\lfloor
    \frac{1-\delta_h/(2N)}{1+\epsilon}
    \del{
      m_h^* - \del{ \frac1{2\epsilon} + \frac13 } \lg(2N^2/\delta_h)\ln(2/\delta_h)
    }
    \right\rfloor
    ,
    \quad
    \delta_e := 2\delta_h ,
  \]
and same, fixed
  communication bound $m$.
\end{claim}

\begin{proof}
  \newcommand\Head{\mathsf{Head}}
  \newcommand\Tail{\mathsf{Tail}}
  Assume $\acs$ is a Huffman-asymmetric-coding scheme with parameters $m_h^*$ and $\delta_h$.
  Take any $\mu \in \Delta([N])$ with $kH(\mu) \leq m_e^*$.
  Define $\delta_0 := \delta_h/(2N^2)$.
  Let $\Head := \{ i \in [N] : \mu(i) \geq \delta_0 \}$ and $\Tail := [N] \setminus \Head$.
  Let $E$ be the event where $S \sim \mu^k$ satisfies $S \subseteq \Head$.
  Since $(1-N\delta_0)^k \geq 1 - Nk\delta_0 \geq 1-\delta_h/2$, it follows that
  \[
    \Pr_{S \sim \mu^k}(E) \geq 1 - \delta_h/2 .
  \]
  Furthermore, conditional on $E$, we can bound the expected value of $\sum_{i \in S} \lg(1/\mu(i))$ as follows:
  \[
    k H_E(\mu)
    := \E_{S \sim \mu^k}\sbr{ \sum_{i \in S} \lg(1/\mu(i)) \;\middle\vert\; E }
    = \frac{k}{1 - \mu(\Tail)} \sum_{i \in \Head} \mu(i) \lg(1/\mu(i))
    \leq \frac{k}{1-\delta_h/(2N)} H(\mu)
    .
  \]
  By Bernstein's inequality, we have
  \[
    \Pr_{S \sim \mu^k}\del{
      \sum_{i \in S} \lg\frac1{\mu(i)}
      \leq k H_E(\mu)
      + \sqrt{2k H_E(\mu) \lg\del{\frac{2N^2}{\delta_h}} \ln\del{\frac{2}{\delta_h}}}
      + \frac{\lg\del{\frac{2N^2}{\delta_h}}\ln\del{\frac{2}{\delta_h}}}{3}
      \;\middle\vert\; E
    } \geq 1-\frac{\delta_h}{2} .
  \]
  Therefore, with probability at least $1-\delta_h$ over the random draw $S \sim \mu^k$, we have
  \begin{align*}
    \sum_{i \in S} \lg(1/\mu(i))
    & \leq
    \frac{k H(\mu)}{1-\delta_h/(2N)} + \sqrt{\frac{2k H(\mu) \lg(2N^2/\delta_h)\ln(2/\delta_h)}{1-\delta_h/(2N)}} + \frac{\lg(2N^2/\delta_h)\ln(2/\delta_h)}{3}
    \\
    & \leq
    \frac{1+\epsilon}{1-\delta_h/(2N)} k H(\mu)
    + \del{ \frac1{2\epsilon} + \frac13 }
    \lg(2N^2/\delta_h)\ln(2/\delta_h)
    \\
    & \leq
    m_h^*
  \end{align*}
  where the second inequality follows from the arithmetic-mean/geometric-mean inequality, and the last inequality uses the definition of $m_e^*$.
  Conditional on this event, $\acs$ correctly decodes the set $S$ with probability at least $1-\delta_h$.
  Thus, $\acs$ is an entropy-asymmetric-coding scheme with parameters $m_e^*$ and $\delta_e = 2\delta_h$.
\end{proof}

\end{document}